\documentclass[10pt,conference]{IEEEtran}

\usepackage{amsmath, amssymb}
\usepackage{epsfig,epsf,pstricks,pgf}
\usepackage{cite}

\usepackage{amsthm}
\usepackage{psfrag}

\newtheorem{theorem}{Theorem}
\newtheorem{lemma}{Lemma}
\newtheorem{cor}{Corollary}

\newtheorem{prop}{Proposition}

\theoremstyle{definition}
\newtheorem{definition}{Definition}

\newcommand{\bE}{\mathbb{E}} 
 
\newcommand{\bY}{\mathbf{Y}} 
\newcommand{\bX}{\mathbf{X}}
 
\newcommand{\bS}{\mathbf{S}} 
 
\newcommand{\by}{\mathbf{y}} 
\newcommand{\bx}{\mathbf{x}}
 
\newcommand{\bs}{\mathbf{s}} 

\newcommand{\bU}{\mathbf{U}}
\newcommand{\bu}{\mathbf{u}}

\newcommand{\bA}{\mathbf{A}}

\newcommand{\bB}{\mathbf{B}}

\newcommand{\one}{\boldsymbol{1}} 
\newcommand{\rw}{\rightarrow}

\begin{document}

% paper title
\title{``Compressed'' Compressed Sensing}
\author{
\authorblockN{Galen Reeves and Michael Gastpar}
\authorblockA{Department of Electrical Engineering and Computer Sciences \\
University of California, Berkeley \\
Berkeley, CA, 94720, USA \\
Email: \{greeves, gastpar\}@eecs.berkeley.edu}
}

% make the title area
\maketitle

\begin{abstract}

The field of compressed sensing has shown that a sparse but otherwise arbitrary vector can be recovered exactly from a small number of randomly constructed linear projections (or samples). The question addressed in this paper is whether an even smaller number of samples is sufficient when there exists prior knowledge about the distribution of the unknown vector, or when only partial recovery is needed. An information-theoretic lower bound with connections to free probability theory and an upper bound corresponding to a computationally simple thresholding estimator are derived. It is shown that in certain cases (e.g. discrete valued vectors or large distortions) the number of samples can be decreased. Interestingly though, it is also shown that in many cases no reduction is possible. %Interestingly, it is shown that in many cases prior knowledge about the distribution does not decrease the number of samples that are needed whereas for partial recovery, the number of samples can be decreased. %Connections with random matrix theory and free probability are established.
%It is shown that in certain cases (e.g. discrete valued vectors or large distortions) the number of samples can be decreased. Interestingly though, it is also shown that in many cases no reduction is possible.
\end{abstract}

%%%%%%%%%%%%%%%%%%%%%%%%%%%%%%%%%%%%%%%%%%%%%%%%%%%%%%%%%%%%%

\section{Introduction}

Suppose that an unknown vector $\bx$ of length $n$ is observed using a set of linear projections $\by = A \bx $ where $A$ is a known ${m \times n}$ sampling matrix. The field of compressed sensing (see references in \cite{riceweb}) has shown that if $\bx$ is sparse (i.e. has a relatively small number of nonzero elements) then exact recovery is possible even if the number of samples $m$ is much less than the vector length $n$. A great deal of work has considered necessary and sufficient conditions on the sampling matrix $A$ with respect to various recovery goals. In particular, much of this work has focused on sufficient conditions for computationally efficient recovery algorithms. 

Typically, the conditions on the sampling matrix are remarkably  general with respect to the unknown vector $\bx$ in the sense that they require no assumptions about the values or locations of the nonzero elements. Moreover, many of the results still apply even if $\bx$ is not actually sparse, but instead has a sparse representation with respect to a known basis. %Thus, compressed sensing offers sharp sampling theorems for the general set of sparse signals. 

%, provided that there exists a known basis $B \in \mathbb{R}^{n \times n}$ in which $\bx$ is known to have a sparse representation. %, , but instead has a sparse representation with respect to a known basis.%, i.e. $\bx = B \tilde{\bx}$ where $\tilde{\bx}$ is sparse and $B \in \mathbb{R}^{n \times n}$ is full rank. 

In many practical situations however, there exists prior knowledge about the values of the nonzero elements. In this paper, we address the extent to which this additional information allows for recovery using an even smaller number of samples than are needed in the general ``compressed sensing'' setting. We focus exclusively on recovery of the support set (i.e. the locations of the nonzero elements) in the high dimensional setting and ask the following two questions:

\begin{itemize}
\item {\em What if we consider approximate support recovery?}. In Section \ref{sec:arb} we show that if the sampling matrix $A$ is designed with knowledge of the basis in which $\bx$ is sparse, then there exists a natural tradeoff between accuracy and the number of samples. Conversely, if the sampling matrix is designed independently of the sparse basis, then no such tradeoff is possible. 
\item {\em What if $\bx$ is a random vector with a known distribution?} If the distribution is discrete, then it is straightforward to see that only one sample is needed. In Section \ref{sec:rand} we consider general distributions, %in the setting where the sampling matrix is designed independently of the sparse basis. Our 
and our main results (Theorems \ref{thm:LB} and \ref{thm:UB}) show that knowledge of the distribution may or may not decrease the number of samples that are needed depending on the desired distortion and various properties of the distribution such as the differential entropy.%, such as the differential entropy, of the distribution. %Our lower bound uses results from random 
% if the normalized differential entropy of the distribution is large, then knowledge of the distribution does not decrease the number of samples that are needed for accurate recovery. However, we also show that if a relatively large distortion is allowed, then knowledge of the distribution allows for recovery using a computationally simple estimator with far few samples than are needed for arbitrary vectors.  
%\item {\em What if $\bx$ is a random vector with a known distribution?} If the distribution is discrete, then is is straightforward to see that only one sample is needed! The general setting, however, is not so straightforward. In Section \ref{sec:rand} we show that if the normalized differential entropy of the distribution is large, then knowledge of the distribution does not decrease the number of samples that are needed for accurate recovery. However, we also show that if a relatively large distortion is allowed, then knowledge of the distribution allows for recovery using a computationally simple estimator with far few samples than are needed for arbitrary vectors.  
\end{itemize}

%The answers to the above questions are interconnected; In many cases knowledge of the distribution helps only if the allowable distortion is large enough. 

An additional contribution of this paper is given by the proof of our main lower bound (Section \ref{sec:proof-LB}) which uses results from free probability theory to characterize the limiting distributions of certain random matrices that occur frequently in compressed sensing.

A number of related works have addressed various bounds on the asymptotic sampling rate needed for the noisy setting \cite{Wainwright_allerton06,Wainwright_isit07,Reeves_masters,aeron-2007,aeron-2008,akcakaya-2007,RG08,fletcher-2008,WWR08,RG-Lower-Bounds,RG-Upper-Bounds }. In these cases, it is clear that properties such as the size of the smallest nonzero values dramatically affect the number of samples that are needed. The noiseless setting addressed in the paper however, gives insight about fundamental limitations of the sampling process that cannot be overcome simply by increasing the signal to noise ratio. 

%The consideration of random source considered in the paper is most closely related to the bounds for the noisy setting given in  \cite{}

% . Our main upper bound corresponds to a computationally simple estimator. 

%Also, there exist connections between 

%The answers to the above questions are interconnected. The proof our the lower bound  form free probability theory 

%[PUT PARAGRAPH HERE ABOUT THE CONNECTIONS BETWEEN THE ABOVE QUESTIONS. POINT OUT THE CONNECTION TO RANDOM MATRIX THOERY AND FREE PROBABILITY, AND MENTION THAT WE CAN ACHIEVE GOOOD RESULTS USING COMPUTATIONALLY SIMPLE ESTIMATORS.]

%[ALSO PUT SOME POINTERS HERE TO RELATED WORK]

%The answers we provide to the above questions are interconnected. Interestingly, it turns out that, unlike the classical compressed sensing setting, wether or not the sampling matrix $A$ is allowed to depend on the sparse basis $B$ strongly affects the rate distortion performance. 

%the answers to the above questions depend significantly on whether or not the choice of sampling matrix is $A$ is allowed to depend on the realization of the bases $B$. 

\section{Problem Setup}
We consider a generalized sparsity model where an unknown vector $\bx\in \mathbb{R}^n$ is assumed to have a sparse representation $\bu \in \mathbb{R}^n$ with respect to a known orthonormal basis $B \in \mathbb{R}^{n \times n}$ given by
\begin{align*}
\bx = B \bu.
\end{align*}
The {\em support} $\bs \subset \{1,2,\cdots,n\}$ is the set of integers indexing the nonzero elements of $\bu$,
\begin{align*}
\bs : = \{i : u_i \ne 0\},
\end{align*}
and the sparsity $k = |\bs|$ is the number of nonzero elements. 

The vector of samples $\by \in \mathbb{R}^m$ is expressed in terms of a sampling matrix $A \in \mathbb{R}^{m \times n}$:
\begin{align*}
\by = A \bx .%= A B \bu.
\end{align*}
Throughout this paper, we assume that an estimator is given the set $(\by,A,B,k)$ and the goal is to recover the support $\bs$ of the sparse representation $\bu$. The distortion between a support $\bs$ and its estimate $\hat{\bs}$ is measured using the Hamming distance
\begin{align*}
d(\bs,\hat{\bs}) := |\bs \cup \hat{\bs}| - |\bs \cap \hat{\bs}|.
%d(\bs,\hat{\bs}) := {\textstyle  \frac{1}{n}}\big[ |\bs \cup \hat{\bs}| - |\bs \cap \hat{\bs}| \big]
\end{align*}

This paper focuses on whether or not a given recovery task is possible using an $m \times n$ sampling matrix $A$. One possible requirement is that $A$ be good uniformly for all possible $k$-sparse vectors. However, this paper considers a less stringent requirement and instead asks if there exists a distribution $p_A$ such that recovery is possible, with high probability,  for any $k$-sparse vector $\bu$ when $\bA \sim p_A$ is a {\em random} matrix drawn independently of $\bu$, and possibly also $B$. 

To highlight the difference between the above requirements it is useful to consider the task of exact recovery. Then, it can be shown that there exists a sampling matrix $A$ satisfying the first requirement if and only if $m \ge \min(2k,n)$, whereas there exists a distribution $p_A$ satisfying the second requirement if and only if $m \ge \min(k+1,n)$. 

To characterize the number of samples that are needed, we focus on the high dimensional setting where the vector length $n$ becomes large. We assume that for each $n$, the sparsity is given by $k_n =  \lfloor \Omega \cdot n \rfloor$ for some known {\em sparsity rate} $\Omega \in (0,1)$. The following definitions are used to characterize the asymptotic {\em sampling rate} given by  $\rho =m_n/n$.%the number of samples that are needed asymptotically. 

\begin{definition} %For each integer $n$ and sparsity rate $\Omega$, 
The {\em general source} $\mathcal{X}^n(\Omega)$ outputs an arbitrary (non-random) vector $\bx \in \mathbb{R}^n$ and basis $B\in \mathbb{R}^{n \times n}$ where $\bx = B \bu$ for some vector $\bu \in \mathbb{R}^n$ whose support $\bs$ has size $k=\lfloor \Omega \cdot n\rfloor$. 
\end{definition}
 
Given any support estimator $\hat{\bs}(\by,A,B,k)$ and any distribution $p_A$, the probability that the fraction of errors exceeds the normalized distortion $\alpha \in [0,1]$ for the general source $\mathcal{X}^n(\Omega)$ is given by
\begin{align*}
%P_e^{(n)} =  \inf_{\bx\; :\; |\bs| = \lfloor \Omega \cdot n\rfloor}\Pr\Big\{ d\big(\hat{\bx}(\bA\bx), \bx\big) > \alpha \Big\}.
%P_e^{(n)} =  \inf_{\bx\; :\; |\bs| = \lfloor \Omega \cdot n\rfloor}\Pr\Big\{ d\big(\hat{\bx}(\bA\bx), \bx\big) > \alpha \Big\}.
P_e^{(n)} =  \inf_{(\bx,B) \in \mathcal{X}^n(\Omega)}\Pr \!\Big\{ d\big(\bs,\hat{\bs}(\by,\bA,B,k),\big) > \alpha \cdot k \Big\}.
%P_e\big(\mathcal{X}^n(\Omega)\big) =  \inf_{(\bx,B) \in \mathcal{X}^n(\Omega)}\Pr \!\Big\{ d\big(\bs,\hat{\bs}(\by,\bA,B,k),\big) > \alpha\! \cdot k\! \Big\}.
%P_e^{(n)} =  \inf_{(\bx,B) \in \mathcal{X}^n(\Omega)}\Pr\Big\{ d\big(\bs,\hat{\bs}(\by,\bA,B,k),\big) > \alpha \cdot k \Big\}.
%P_e^{(n)} =  \inf_{\bx \in \mathcal{X}^n(\Omega)}\Pr\Big\{ d(\hat{\bX}, \bx) > \alpha \Big\}.
\end{align*}

\begin{definition}\label{def:rho_a}
A sampling rate distortion pair $(\rho,\alpha)$ is said to be {\em achievable} for a source $\mathcal{X}$ if for each integer $n$ there exists an estimator $\hat{\bs}(\by,A,B,k)$ and a distribution $p_A$ on a $\lceil \rho \cdot n \rceil \times n$ sampling matrix such that
%
%there exists a sequence of distributions on 
%
% for each $n$, there exists a distribution on the sampling matrix $\bA\in \mathbb{R}^{\lceil \rho \cdot n \rceil \times n}$, and an estimator $\hat{\bx}(\by;A,B) : \mathbb{R}^{\lceil \rho \cdot n \rceil } \mapsto \mathbb{R}^n$ such that 
\begin{align*}
P_e^{(n)} \rw 0 \quad \text{as} \quad n \rw \infty.
\end{align*}
The {\em sampling rate distortion function} $\rho(\alpha)$ is the infimum of rates $\rho \ge 0$ such that the pair $(\rho,\alpha)$ is achievable.
\end{definition}

\section{Arbitrary Signals}\label{sec:arb}

This section considers the sampling rate distortion function $\rho(\alpha)$ of the general source $\mathcal{X}(\Omega,F)$ for two different restrictions on the sampling matrix.

\begin{definition}
A random sampling matrix $\bA$ is said to be {\em universal} if it is drawn independently of the basis $B$, and {\em basis-specific} otherwise. 
\end{definition}

One useful property of a universal sampling matrix is that the sampling matrix can be constructed without knowledge of the sparse basis. Recovery with respect to a basis-specific matrix, however,  is equivalent to assuming the the basis is the identity matrix (i.e. $B=I$) since any target matrix $\bA_0$ designed for this setting can be applied to a general basis $B$ by using the sampling matrix $\bA = \bA_0 B^{-1}$. The following result shows that the universal and basis-specific settings are the same when exact recovery is required but significantly different when a nonzero distortion is allowed. 

%while conditions for exact recovery are the same in both settings, the difference for nonzero distortions can be significant. 

%If the sampling matrix is un

%

%This section 

%The setting 

%[DISCUSS HERE THE DIFFERENCE BETWEEN THE SETTINGS WHERE THE THE DISTRIBUTION ON THE SAMPLING MATRIX CAN AND CANNOT DEPEND ON THE SPARSE BASIS. I STILL NEED GOOD NAMES FOR THESE SETTINGS]

%To begin our investigation, we discuss the implications of whether or not the basis is known a priori. For instance, we might use the notation $$ \by = \Phi \Psi\bx$$ where $\Psi$ is distributed uniformly over the set of orthonormal matrices. We ask, does there exist a distribution on $\Phi$, that is independent of the basis $\Psi$? This is equivalent to saying that the distribution on $\bA$ is {\em rotationally invariant}, i.e. $\bA$ is equal in distribution to $\bA U$ for any orthonormal matrix $U$. 

\begin{prop}\label{prop:arb}The sampling rate distortion function $\rho(\alpha)$ of the general source $\mathcal{X}(\Omega)$ is given by
%If the nonzero elements are chosen arbitrarily, then the sampling rate distortion function $\rho(\alpha)$ is given by
\begin{align}
\rho(\alpha) =
\begin{cases}
\Omega, & \text{\normalfont if $\bA \sim p_{A}$ is universal} \\
(1-\alpha)\Omega, & \text{\normalfont if $\bA  \sim p_{A|B}$  is basis-specific}\\
\end{cases}
\end{align}
for $\alpha < 1$ and is equal to zero otherwise. %$\rho(\alpha) = 0$ if $\alpha =0$.
%%if the basis is known, and
%\begin{align*}
%\rho(\alpha) = \Omega.
%\end{align*}
\end{prop}
\begin{proof}[Proof Sketch]
If the basis is known, then a ``rate sharing'' strategy may be employed to convexify the achievable rate distortion region. Roughly speaking, this corresponds to ignoring some randomly chosen subset of the elements of $\bu$ by placing zeros in the corresponding columns of the matrix $AB$. In the universal setting, however, this strategy is not possible. A full proof is given in \cite{RG-Lower-Bounds}.
\end{proof}

%These bounds are illustrated in Figure \ref{fig:arbX} below.

%[DISCUSS THE ABOVE RESULT HERE AND FIGURE \ref{fig:arbX}. MAKE SOME REMARKS ABOUT THE FACT THAT KNOWING THE SPARSE BASIS ALLOWS FROM RATE SHARING. NOTE THE THE PROOFS ARE ESSENTIALLY ALGEBRAIC.]

%The 
%The behavior of the above result is shown in Figure \ref{fig:arbX}. The conclusions are, that again this story is highly algebraic. The only twist provided by the distortion, is that randomized ``times sharing'' between different rate points is possible, and thus the achievable region is convex. 

\begin{figure}[htbp]
\centering
\psfrag{x1}[]{ \small  $\alpha$}
\psfrag{y1}[B]{\small $\rho(\alpha)/\Omega$} %(no. samples / signal length)}
\psfrag{t1}[B]{\small }%\footnotesize Distortion $\alpha = 0.1$}
\psfrag{label two}[b]{\small $\qquad \qquad$ Basis-Specific: $\bA \sim p_{A|B}$}
\psfrag{label one}[b]{\small Universal: $\bA \sim p_{A}$}
\vspace{-.2cm}
\epsfig{file=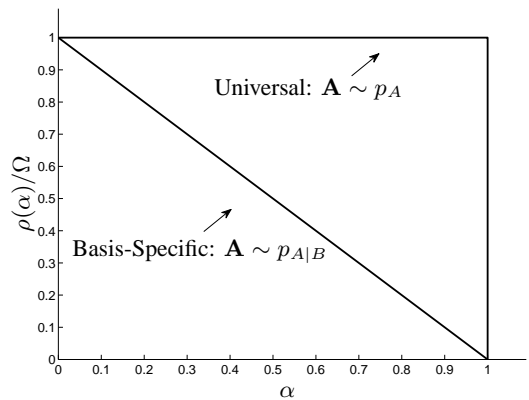, width = .79\columnwidth}
\vspace{-.2cm}
\caption{Comparison of the normalized sampling rate distortion function $\rho(\alpha)/\Omega$ of the general source $\mathcal{X}(\Omega)$ as a function of the distortion $\alpha$ for the universal and basis-specific settings.
}
\label{fig:arbX}
\end{figure}

%\begin{prop}[Arbitrary Signals, Known Basis] If the nonzero elements are chosen arbitrarily, and the distribution sampling matrix $\bA$ is allowed to depend on the realization of the sparse basis $\bB$, then the sampling rate distortion function $\rho(\alpha)$ is given by
%%\begin{align}
%%\rho(\alpha) =
%%\begin{cases}
%%(1-\alpha)\Omega, & \text{\normalfont if n} \\
%%\Omega, & \text{\normalfont if basis is random}\\
%%\end{cases}
%%\end{align}
%%if the basis is known, and
%\begin{align*}
%\rho(\alpha) = (1-\alpha) \Omega.
%\end{align*}
%\end{prop}
%\vspace{3in}

\section{Random Signals}\label{sec:rand}

So far, we have considered the recovery of arbitrary vectors and the results have been mostly algebraic. In this section, we consider recovery of random vectors. We focus exclusively on the universal setting where the sampling matrix $\bA$ must be designed independently of the sparse basis $B$.

\begin{definition} %For each integer $n$, sparsity rate $\Omega$ and distribution function $F$ with finite power and zero probability mass a the origin, 
The {\em random source} $\mathcal{X}^n(\Omega,F)$ outputs a random vector $\bX \in \mathbb{R}^n$ and basis $\bB \in \mathbb{R}^{n \times n}$ where $\bX = \bB \bU$ for a random vector $\bU \in \mathbb{R}^n$ whose support $\bS$ is distributed uniformly over all possibilities of size $k=\lfloor \Omega \cdot n\rfloor$ and whose nonzero elements $\{ U_i : i \in \bS\}$ are i.i.d. $\sim F$. The basis $\bB$ is distributed uniformly over the set of all orthonormal matrices and is independent of $\bU$. %The support $\bS$ distributed uniformly over all subsets of $\{1,2,\cdots,n\}$ of size $\lfloor \Omega \cdot n\rfloor$ and the nonzero elements $\{ U_i : i \in \bS\}$ are independently and identically distributed with distribution function $F$. 
\end{definition}

We assume throughout that $F$ denotes the distribution of a real valued random variable with finite power and zero probably mass at zero. Also, the definitions of achievability are the same as for the general source, except that the probability of error is taken with respect to the random vector $\bX$ and random basis $\bB$,
%Given any support estimator $\hat{\bs}(\by,A,B,k)$ and any distribution $p_A$ on $\bA$, the probability  that the fraction of errors exceeds the normalized distortion $\alpha \in [0,1]$ for the random source $\mathcal{X}^n(\Omega,F)$ is given by
%Given any distribution $p_A$ on a random sampling matrix $\bA\in \mathbb{R}^{m \times n}$, any support estimator $\hat{\bs}(\by,A,B.k)$, and any distortion $\alpha\in [0,1]$, the probability of erroneous recovery for the general source $\mathcal{X}(\Omega)$ is given by 
\begin{align}
%P_e^{(n)} =  \inf_{\bx\; :\; |\bs| = \lfloor \Omega \cdot n\rfloor}\Pr\Big\{ d\big(\hat{\bx}(\bA\bx), \bx\big) > \alpha \Big\}.
%P_e^{(n)} =  \inf_{\bx\; :\; |\bs| = \lfloor \Omega \cdot n\rfloor}\Pr\Big\{ d\big(\hat{\bx}(\bA\bx), \bx\big) > \alpha \Big\}.
P_e^{(n)} = \Pr \!\Big\{ d\big(\bS,\hat{\bs}(\bY,\bA,\bB,k)\big) > \alpha \cdot k \Big\}.\label{eq:PeRandom}
%P_e\big(\mathcal{X}^n(\Omega)\big) =  \inf_{(\bx,B) \in \mathcal{X}^n(\Omega)}\Pr \!\Big\{ d\big(\bs,\hat{\bs}(\by,\bA,B,k),\big) > \alpha\! \cdot k\! \Big\}.
%P_e^{(n)} =  \inf_{(\bx,B) \in \mathcal{X}^n(\Omega)}\Pr\Big\{ d\big(\bs,\hat{\bs}(\by,\bA,B,k),\big) > \alpha \cdot k \Big\}.
%P_e^{(n)} =  \inf_{\bx \in \mathcal{X}^n(\Omega)}\Pr\Big\{ d(\hat{\bX}, \bx) > \alpha \Big\}.
\end{align}
%The achievability of a sampling rate distortion pair $(\rho,\alpha)$ and the sampling rate distortion function $\rho(\alpha)$ are the same as before and are given by Definition \ref{def:rho_a}.

%We define $\mathcal{F}_0$ to be the set of all cumulative distribution functions (henceforth referred to simply as distributions) $F(x) = \Pr\{X \le x\}$  for a random variable $X$ such that $\bE X^2 < \infty$ and $\Pr\{X=0\}=0$. Note that the set $\mathcal{F}_0$ contains discrete, continuous, and mixed distributions. 

\subsection{Lower Bounds}
This section gives an information theoretic lower bound on the sampling rate distortion function $\rho(\alpha)$ of a random source $\mathcal{X}(\Omega,F)$. To begin, we note that in some cases, the constraints imposed by the distribution $F$ significantly alter the nature of the estimation problem.

\begin{prop}[Discrete Signals]\label{ex:discrete}
Suppose that the distribution $F$ is supported on a discrete and finite set $\Sigma \subset \mathbb{R}\backslash\{0\}$. Then, only $m=1$ sample is sufficient for exact recovery, and the sampling rate distortion function $\rho(\alpha)$ of the random source $\mathcal{X}(\Omega,F)$ is $\rho(\alpha) = 0$ for all $\alpha$.
\end{prop}
\begin{proof}%[Proof Sketch]
Suppose that $\bA$ is an $1 \times n$ ``matrix'' whose elements are drawn i.i.d. from continuous distribution with finite power. Then, with probability one, the projection $\bu \mapsto \bA\bB \bu$ maps each of the ${n \choose k}k^{|\Sigma|} $ possible realizations of $\bu$ to a unique real number. 
\end{proof}

The fact that only one sample is needed for discrete distributions is not due to the sparsity in the problem (after all, the result does not depend on the sparsity rate $\Omega$) %, but is instead a consequence of the fact that we are taking a linear combination of numbers drawn from a finite set using real valued coefficients. Hence, Example \ref{ex:discrete} 
and Proposition \ref{ex:discrete} provides little insight into cases where the unknown signal may have a density. To address these cases, we introduce the following property of a random signal source.% $\mathcal{X}(\Omega,F)$.% that, roughly speaking, says have far $F$ is from begin discrete.

\begin{definition}\label{def:Theta}
Given any distribution $F$ with a density and any sparsity rate $\Omega$ the function $\theta(\Omega,F) \in [0,1]$ is given by
\begin{align}
\theta(\Omega,F) & = \frac{ (2\pi e)^{-1} \exp(2 h(F))}{\sigma_F^2 + (1- \Omega)\,  \mu_F^2},
\end{align}
where $\mu_F$, $\sigma_F^2$, and $h(F)$ denote the mean, variance and differential entropy of the distribution $F$. If $F$ does not have a density, then $\theta(\Omega,F) \equiv 0$. 
\end{definition}

The property $\theta(\Omega,F)$ is the normalized entropy power of the nonzero elements and is equal to one if and only if $F$ is a zero mean Gaussian distribution. Roughly speaking, one may interpret $\theta(\Omega,F)$ as the relative ``distance'' between a random source $\mathcal{X}(\Omega,F)$ and a discrete source. The following result, which is proved in Section \ref{sec:proof-LB}, uses this property to lower bound the sampling rate distortion function.

\begin{theorem}[Lower Bound]\label{thm:LB}
A sampling rate distortion pair $(\rho,\alpha)$ is not achievable for the random source $\mathcal{X}(\Omega,F)$  if $\rho < \Omega$ and
%Suppose that the nonzero elements are i.i.d. $\sim F$. Then, a sampling rate distortion pair $(\rho,\alpha)$ is not acheivable
\begin{align}
\frac{\rho}{2} \log\left( \frac{1}{\theta(\Omega,F)} \cdot \frac{\Delta(\rho)}{\Delta(\rho/\Omega)}\right) < H(\Omega) - H(\alpha \Omega)%2 R(\Omega,\alpha) 
\end{align}
where $\theta(\Omega,F)$ is given by Definition \ref{def:Theta}, $H(p) = -p\log(p) - (1-p)\log(1-p)$ is binary entropy and 
\begin{align}
%R(\Omega,\alpha) & = H(\Omega) - H(\alpha \Omega)\\
\Delta(r) &=
\begin{cases}
\left({1-r}\right)^{1-1/r}  &\text{\normalfont if}\; r<1\\
 1&\text{\normalfont if}\; r = 1%\\
%\left(1-1/r\right)^{1-r}  &\text{\normalfont if}\; r>1
\end{cases}. \label{eq:Delta}
%Theta(\Omega,F) & = \frac{ (2\pi e)^{-1} \exp(2 h(F))}{\sigma_F^2 + (1- \Omega)\,  \mu_F^2},
\end{align}
%where $H(p) = -p\log(p) - (1-p)\log(1-p)$ is binary entropy, 
%and $\mu_F$, $\sigma_F^2$, and $h(F)$ denote the mean, variance and differential entropy of the distribution $F$.
%\begin{align}
%\Theta(\Omega,F) %&= \frac{\Omega N(U_1)}{\text{\normalfont var}(X_1)}\\
%& = \frac{ (2\pi e)^{-1} \exp(2 h(F))}{1- \Omega\,  \mu_F^2 / \sigma_F^2},
%\end{align}
%where $\mu_F$, $\sigma_F^2$, and $h(F)$ denote the mean and variance and differential entropy of the distribution $F$, and
%\begin{align}
%\Delta(r) &=
%\begin{cases}
%\left({1-r}\right)^{1-1/r}  &\text{\normalfont if}\; r<1\\
% 1&\text{\normalfont if}\; r = 1\\
%\left(1-1/r\right)^{1-r}  &\text{\normalfont if}\; r>1
%\end{cases}.
%\end{align}
\end{theorem}
%\begin{theorem}[Random Signals: Lower Bound]\label{thm:LB}
%A sampling rate distortion pair $(\rho,\alpha)$ is not achievable for the random source $\mathcal{X}(\Omega,F)$  if $\rho < \Omega$ and
%%Suppose that the nonzero elements are i.i.d. $\sim F$. Then, a sampling rate distortion pair $(\rho,\alpha)$ is not acheivable
%\begin{align}
%\rho \log\left( \frac{1}{\Theta(\Omega,F)} \cdot \frac{\Delta(\rho)}{\Delta(\Omega/\rho)}\right) <2 R(\Omega,\alpha) 
%\end{align}
%with
%\begin{align}
%R(\Omega,\alpha) & = H(\Omega) - H(\alpha \Omega)
%\end{align}
%where $H(p) = -p\log(p) - (1-p)\log(1-p)$ is binary entropy, 
%\begin{align}
%\Theta(\Omega,F) %&= \frac{\Omega N(U_1)}{\text{\normalfont var}(X_1)}\\
%& = \frac{ (2\pi e)^{-1} \exp(2 h(F))}{1- \Omega\,  \mu_F^2 / \sigma_F^2},
%\end{align}
%where $\mu_F$, $\sigma_F^2$, and $h(F)$ denote the mean and variance and differential entropy of the distribution $F$, and
%\begin{align}
%\Delta(r) &=
%\begin{cases}
%\left({1-r}\right)^{1-1/r}  &\text{\normalfont if}\; r<1\\
% 1&\text{\normalfont if}\; r = 1\\
%\left(1-1/r\right)^{1-r}  &\text{\normalfont if}\; r>1
%\end{cases}.
%\end{align}
%\end{theorem}

%The proof of Theorem \ref{thm:LB}, which is given in Section \ref{sec:proof-LB}, depends on basic information inequalities as well as the relationship between the limiting empirical distributions of the singular values of the random matrices $\bA\bB$ and $\bA\bB_{\bS}$, which can be conveniently described using free probability \cite{Voicu83}.

One consequence of Theorem \ref{thm:LB}, is that there is a simple test to see whether or not the sampling rate needed for a random source $\mathcal{X}(\Omega,F)$ is any less than that needed for the general source $\mathcal{X}(\Omega)$.

\begin{cor}[Theorem \ref{thm:LB}] The sampling rate distortion function $\rho(\alpha)$ of the random source $\mathcal{X}(\Omega,F)$ is given by $\rho(\alpha) = \Omega$ for all $\alpha<1$ such that
\begin{align}\label{eq:corthm1}
\theta(\Omega,F) > \Delta(\Omega) \exp\left(-{ \textstyle \frac{2}{\Omega}} \big[ H(\Omega) - H(\alpha \Omega)\big]\right).
%\frac{\Omega}{2} \log\left( \frac{1}{\Theta(\Omega,F)} \cdot \Delta(\Omega)\right) < H(\Omega)
\end{align}
\end{cor}

\subsection{Upper Bounds}
%Hence, 
Theorem \ref{thm:LB} shows that in many cases the sampling rate distortion function of a random source is equal to that of the arbitrary source. However, if $\theta(\Omega,F)$ is less than the right hand side of \eqref{eq:corthm1}, then the lower bound in Theorem \ref{thm:LB} is less than the sparsity rate $\Omega$ and there exists a gap with the upper bound given by the arbitrary setting (Proposition \ref{prop:arb}). In this section, we investigate improved (i.e. lower) upper bounds for these settings. 

One way to upper bound $\rho(\alpha)$ is to directly analyze the estimator that minimizes the error probability $P_e^{(n)}$ given in \eqref{eq:PeRandom}. Although non-asymptotic properties of optimal estimation in the Gaussian setting have been studied (see for example \cite{RG_asilomar09}), analysis in the asymptotic setting appears to be challenging.

%One way to upper bound $\rho(\alpha)$ is to analyze the Bayes optimal estimator, i.e. the estimator that minimizes the error probability $P_e^{(n)}$ given in \eqref{eq:PeRandom}. Although non-asymptotic aspects of the Bayes estimator for Gaussian distributions have been studied (see for example \cite{RG_asilomar09}), the direct analysis of Bayes estimators in the asymptotic setting appears to be challenging. 

%analyze the Bayes optimal estimate of $\bS$ given $\bY$. Although non-asymptotic aspects of the optimal estimate for Gaussian distributions have been analyzed (see for example \cite{asilomar}), the asymptotic behavior for general distributions is, to our knowledge, unknwon. 

In this paper, we instead derive upper bounds for a computationally simple, and potentially suboptimal, estimator described below.

\begin{definition}\label{def:hypo}
Suppose that the distribution of a random variable $X$ is given by
\begin{align*}
X \sim 
\begin{cases}
W, & \text{\normalfont if $Z=0$}\\
W+\sqrt{\rho}\,U, & \text{\normalfont if $Z=1$}
\end{cases}
\end{align*}
where $U\sim F$, $W \sim \mathcal{N}(0, \Omega\, \bE[U^2] )$, and $Z \sim\text{Bernoulli}(\Omega)$ are independent. For any subset $T \subseteq \mathbb{R}$ let $\hat{Z}_T(x) = \one(x \in T)$ and define the error probability %for estimate of $Z$ given $X$ to be
%\begin{align}
%\hat{Z}_T(x)& = 
%\begin{cases}
%0 ,& \text{if $x \notin T$}\\
%1, & \text{if $x\in T$}
%\end{cases}%\\
%T^* &= \arg \inf_T \Pr\{ \hat{Z}_T (X)\ne Z\}
%\end{align}
\begin{align}
%T^* &= \arg \inf_{T \subseteq \mathbb{R}} \Pr\{ \hat{Z}_T (X)\ne Z\}\\
\epsilon(\rho,\Omega,F) & = \inf_{T \subseteq \mathbb{R}} \Pr\{ \hat{Z}_T (X)\ne Z\} \label{eq:eps}.%\Pr\{ \hat{Z}_{T^*} (X)\ne Z\}.
\end{align}

%and   %$W\sim \mathcal{N}(0,1)$ and independent of $U$ and 
%$\Pr(H_1)=1-\Pr(H_2) = \Omega$.
\end{definition}

%asymptotic upper bounds on the sampling rate needed for such estimation are unknown. 
%Although optimal Bayesian estimators for specific distributions can be analyzed (see for example \ref{asilomar})
%To out knowledge, however, there exist no upper bounds on the asymptotic sampling rate needed for these estimators. 

% it is therefore natural to ask if stronger (i.e. lower) upper bounds can be attained for these settings. Although optimal Bayesian estimators for specific distributions can be analyzed (see for example \ref{asilomar}), to our knowledge there exists no corresponding upper bounds on the sampling rate distortion region comparable to the lower bounds in Theorem \ref{thm:LB}.  

%One question then, is whether one can show an improved upper bound for these 
%in these setting that . Although 

\begin{definition}% Given any set $(\by,A,B,k)$, 
For a random source $\mathcal{X}(\Omega,F)$, the {\em Thresholding} (TH) estimator $\hat{\bs}_{\text{\normalfont TH}}(\by)$ is given by
\begin{align}
\hat{\bs}_{\text{\normalfont TH}}(\by)& = \big\{i \; : \;  \hat{\bu}_i \in T^* \big\}
\end{align}
where $\hat{\bu} = B^TA^T \by \in \mathbb{R}^n$ and $T^*\subseteq \mathbb{R}$ minimizes the right hand side of \eqref{eq:eps} with $\rho = m/n$.% when $\Omega = k/n$ and $\rho = m/n$. 
%is the minimal probability of error in the hypothesis testing problem where the distribution of an observed random variable $X$ is given by%given by %=\inf_{\hat{H}} \Pr(\hat{H} \ne H |X)$ where
% and $\Pr(H_2)= \Omega$. 
%\begin{align}
%T = \big\{ x \in \mathbb{R} \; : \; Q(x|H_1)/Q(x|H_2) > n/k-1  \big \}
%\end{align}
%with $Q_1$ and $Q_2$ denotes the probability density 
\end{definition}

The thresholding estimate corresponds to a separate hypothesis test for each element of $\bx$ and its complexity is linear in the vector length $n$. %Its asymptotic performance with respect to an i.i.d. Gaussian distribution on the sampling matrix is characterized by the following result.% for an  i.i.d. Gaussian distribution on $\bA$. 

\begin{prop}\label{prop:TH}
Suppose that for each integer $n$, the elements of the sampling matrix $\bA$ are i.i.d. $\sim \mathcal{N}(0,1/n)$. Then, for any random source $\mathcal{X}(\Omega,F)$ and sampling rate $\rho$, %the distortion of thresholding estimator $\hat{\bs}_\text{\normalfont TH}$ is asymptotically equal to error probability of the hypothesis test given in Definition \ref{def:hypo},
\begin{align*}
d(\hat{\bS}_\text{\normalfont TH}, \bS) \rw \epsilon(\rho,\Omega,F) \; \text{\normalfont in probability as $\; n \rw \infty$}
\end{align*}
where $\epsilon(\rho,\Omega,F)$ is given by \eqref{eq:eps}. 
%Suppose that for each $n$, the sampling matrix $\bA$ is is bla bla bl a
\end{prop}
\begin{proof}[Proof Sketch] The key step, which is proved in \cite{RG-Upper-Bounds}, is to show that the empirical distributions of $\{\hat{U}_i, i \in \bS\}$ and $\{\hat{U}_i, i \notin \bS\}$ converge to the distribution of the random variable $X$ described in Definition 6 conditioned on the events $Z=1$ and $Z=0$ respectively. %This convergence is proved in %\ref{
\end{proof}

Combining Propositions \ref{prop:arb} and \ref{prop:TH} gives the following result which is complementary to Theorem \ref{thm:LB}.

%By analyzing the performance 

%For each integer $n$, let $\hat{\bX}_n = $\

%In the paper (bla bla bla) is shown that the empirical distribution functions 

%The computational complexity of the thresholding estimator is linear in the vector length $n$. 

%[MOTIVATE AND DISCUSS THE UPPER BOUND. PERHAPS DESCRIBE THE THRESHOLDING/HYPOTHESIS TESTING ALGORITHM FIRST TO MAKE THE STATEMENT OF THE THEOREM LESS CONVOLUTED]

\begin{theorem}[Upper Bound]  \label{thm:UB}
A sampling rate distortion pair $(\rho,\alpha)$ is achievable for the random source $\mathcal{X}(\Omega,F)$ if $\rho > \Omega$ or
%Suppose that the nonzero elements are i.i.d. $\sim U$ with $\bE U^2=1$. Then, a sampling rate distortion pair $(\rho,\alpha)$ is acheivable if either $\rho > \Omega$ or
%\begin{align}
$\alpha\Omega > \epsilon(\rho,\Omega,F) $
%\end{align}
where $\epsilon(\rho,\Omega,F)$ is given by \eqref{eq:eps}. %$P_e= \bE\left[ \min_{\hat{H}\in \{1,2\}} \Pr(\hat{H} \ne H |X) \right]$ 
%is the minimal probability of error in the hypothesis testing problem where the distribution of an observed random variable $X$ is given by%given by %=\inf_{\hat{H}} \Pr(\hat{H} \ne H |X)$ where
%\begin{align}
%X \sim 
%\begin{cases}
%\mathcal{N}(0,1)+\sqrt{\frac{\rho}{\Omega}}U, & \text{\normalfont under hypothesis $H=1$}\\
%\mathcal{N}(0,1), & \text{\normalfont under hypothesis $H=2$}
%\end{cases}
%\end{align}
%with %$W\sim \mathcal{N}(0,1)$ and independent of $U$ and 
%$\Pr(H=1)=1-\Pr(H=2) = \Omega$.% and $\Pr(H_2)= \Omega$. 
\end{theorem}

\subsection{A Gaussian Example}
%[THIS SECTION IS NOT DONE YET AND NEEDS TO BE MADE CONCISE]	

This section illustrates the bounds in Theorems \ref{thm:LB} and \ref{thm:UB} for a random source $\mathcal{X}(\Omega,F)$ where $F$ is a Gaussian distribution with mean $\mu$ and variance $1-\mu^2$. In Figure \ref{fig:EX1},  the normalized sampling rate distortion function $\rho(\alpha)/\Omega$ is plotted as a funciton of the mean $\mu$  for  $\alpha = 0.3$. It is shown that if $\mu \le \mu^* \approx 0.83$ then the number of samples needed is no different than for the arbitrary source $\mathcal{X}(\Omega)$.  signals. However, if $\mu > \mu^*$, there exists a gap between the bounds. 

In Figure \ref{fig:EX2}, the same bounds are shown for the relatively large distortion $\alpha = 0.95$. In this case, the upper bound from Theorem \ref{thm:UB} is less than the rate needed for the arbitrary source, which verifies that, in some cases, there is a reduction in the number of samples that are needed. We note that the special case $\mu=1$ corresponds to a discrete distribution, and thus $\rho(\alpha)=0$ by Proposition \ref{ex:discrete}.

\begin{figure}[htbp]
\centering
\psfrag{x1}[]{\small $\mu$}
\psfrag{y1}[B]{\small  $\rho/\Omega$} %(no. samples / signal length)}
\psfrag{t1}[Bt]{\small Moderate Distortion ($\alpha = 0.3$)}%\footnotesize Distortion $\alpha = 0.1$}
%\psfrag{MC Upper Bound [?]}[lb]{\small TH Upper Bound \cite{RG-Upper-Bounds}}
%\psfrag{ML Upper Bound [?]}[lb]{\small NS Upper Bound \cite{RG-Upper-Bounds}}
%\psfrag{OPT Lower Bound [?]}[lb]{\small OPT Lower Bound \cite{RG-Lower-Bounds}}
\psfrag{(This Paper)}[lT]{\small (This Paper)}
\epsfig{file=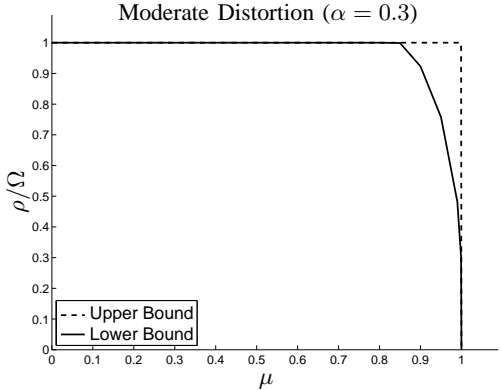, width = .75\columnwidth}
\vspace{-.2cm}
\caption{\label{fig:EX1}
Bounds on the normalized sampling rate $\rho/\Omega$ needed to achieve distortion $\alpha = 0.3$ as a function of the distribution mean $\mu$ when the sparsity rate is $\Omega = 0.35$ and the nonzero signal elements are i.i.d. $\mathcal{N}(\mu,1-\mu^2)$.
}

\end{figure}

\begin{figure}[htbp]
\centering
\psfrag{x1}[]{ \small  $\mu$}
\psfrag{y1}[B]{\small $\rho/\Omega$} %(no. samples / signal length)}
\psfrag{t1}[Bt]{\small High Distortion ($\alpha = 0.95$)}%\footnotesize Distortion $\alpha = 0.1$}
%\psfrag{MC Upper Bound [?]}[lb]{\small TH Upper Bound \cite{RG-Upper-Bounds}}
%\psfrag{ML Upper Bound [?]}[lb]{\small NS Upper Bound \cite{RG-Upper-Bounds}}
%\psfrag{OPT Lower Bound [?]}[lb]{\small OPT Lower Bound \cite{RG-Lower-Bounds}}
\psfrag{(This Paper)}[lT]{\small (This Paper)}
\epsfig{file=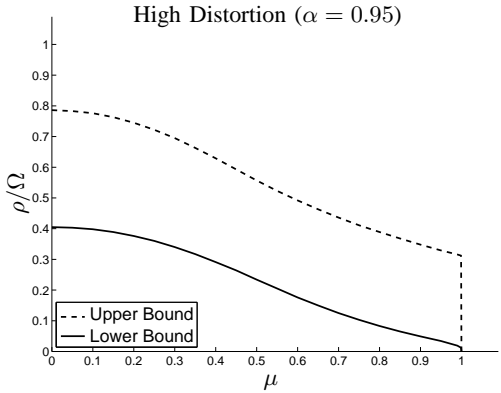, width = .75\columnwidth}
\vspace{-.2cm}
\caption{ \label{fig:EX2} Bounds on the normalized sampling rate $\rho/\Omega$ needed to achieve distortion $\alpha = 0.95$ as a function of the distribution mean $\mu$ when the sparsity rate is $\Omega = 0.35$ and the nonzero signal elements are i.i.d. $\mathcal{N}(\mu,1-\mu^2)$.
}
\end{figure}

%bla bla

%\newpage

%\appendix

\newpage

\section{Proof of Theorem \ref{thm:LB}}\label{sec:proof-LB}

Throughout this proof we use the notation $A$ interchangeably to denote either a particular $m_n \times n$ matrix $A_n$ or a sequence of such matrices $\{A_n\}$. We begin with the following lemma which shows that the sampling rate distortion function can be lower bounded by considering an arbitrary sequence $A$. 

%by consideri

\begin{lemma}
Let $A$ denote any sequence of full rank $\lceil \rho \cdot n \rceil \times n$ sampling matrices. Then, for any distortion $\alpha$, the sampling rate distortion pair $(\rho,\alpha)$ is not achievable for the random source $\mathcal{X}(\Omega,F)$ if 
\begin{align}\label{eq:fano}
\limsup_{n \rw \infty}  {\textstyle \frac{1}{n}}  I(A \bX;\bS |\bB) < H(\Omega) - H(\alpha \Omega).
\end{align}
\end{lemma}
\begin{proof}[Proof Sketch]
The lower bound for a given sequence $A$ follows from Fano's inequality (see e.g. \cite{RG-Lower-Bounds}). The fact that the bound for one matrix $A$ applies to any other matrix $A'$ (of equal rank) follows from that fact that there exists an invertible matrix $D$ (based on the singular value decomposition) such that $D A\bX$ is equal in distribution to $A' \bX$.% bijection between $A\bX$ 
\end{proof}

Next, we upper bound the left hand side of \eqref{eq:fano}. %ull rank sampling matrix $A$. 
Expanding the mutual information for a given problems size $n$ gives
\begin{align*}
I(A \bX; \bS|\bB) &= h(A\bX |\bB) - h(A\bX| \bS,\bB).
\end{align*}
The entropy $h(A \bX |\bB)$ is upper bounded by the entropy of a Gaussian vector with the same covariance as $A \bX$, %\cite{elementsofIT},
 and thus %whose covariance is the same as $\bY$ and thus
%, the first term on the right obyes
%Using the fact that the Gaussian distribution maximizes differential entropy, the first term is upper bounded by
%Since the Gaussian distribution maximize differential entropy, the first term on the right is upper bounded by 
\begin{align*}
h(A \bX |\bB) \le {\textstyle \frac{m}{2}} \log\big(2 \pi e\, \sigma^2_x |{A} {A}^T|^{\frac{1}{m}}\big)
\end{align*}
where $\sigma^2_x = \Omega \sigma_F^2 + \Omega(1-\Omega)\mu_F^2$ is the variance of each element of $\bX$. Furthermore, the entropy $h(A \bX|\bS,\bB)$ is lower bounded by the entropy power inequality \cite{elementsofIT} as
%, the conditional entropy $h(\bY|\bS,\bA,\bB)$ is lower bounded by
%using the entropy power inequality,  the entropy $h(\bY|\bS,\bA,\bB)$ can be lower bounded using using the entropy power inequality to , the second term on the right is lower bounded by
\begin{align*}
h(A \bX |\bS,\bB) \ge {\textstyle \frac{m}{2}} \bE\log\big(2 \pi e\, N(F) |A\bB_\bS \bB^T_\bS A^T|^{\frac{1}{m}}\big)
\end{align*}
where $N(F) =  (2\pi e)^{-1} \exp(2 h(F))$ is the entropy power of each nonzero element of $\bU$. Combining these bounds gives
\begin{align*}
I(A \bX;\bS|\bB) \le {\textstyle \frac{m}{2}} \bE \log\left(\frac{1}{\theta(\Omega,F)} \cdot \frac{|{A} {A}^T|^{\frac{1}{m}}}{ | \frac{1}{\Omega}A\bB_\bS \bB^T_\bS A^T|^{\frac{1}{m}}} \right)
\end{align*}
where we use the fact that $\theta(F,\Omega) = \Omega N(F)/\sigma_x^2$. 

%Although the mutual information $I(A \bX;\bS|\bB)$ is equal for all full rank matrices
Without any loss of generality, we may assume that the spectral distribution of $AA^T$ converges to a compactly supported probability measure $\mu$ as $n \rw \infty$. Then, $|{A} {A}^T|^{\frac{1}{m}} \rw G_\mu$ as $n \rw \infty$ where
\begin{align*}
G_\mu = \int_\mathbb{R} \log(x) d\mu(x).
%G_\mu = \exp \left( \int_\mathbb{R} \log(x) d\mu(x) \right)
\end{align*}

The remaining problem, therefore, is to characterize the spectral distribution of the random matrix $A\bB_\bS \bB^T_\bS A^T$ as $n$ becomes large. To this end, it is convenient to use results from free probability theory which is a theory for non-commutative probability theory developed by Voiculescu \cite{Voicu83}. To begin, observe that the limiting spectral distribution of $A^T A$ has a point mass $\delta_0$ of weight $1-\rho$ at zero and is given by
\begin{align*}
\tilde{\mu} &= (1-\rho) \delta_0 + \rho \mu.
\end{align*}
Observe also, that the limiting spectral distribution of $\bB^T_\bS \bB_\bS$ is given by $$\mu' = (1-\rho/\Omega) \delta_0 + (\rho/\Omega) \delta_1$$
The basic idea from free probability is that the sequences $A^T A$ and $\bB^T_\bS \bB^T_\bS$ are {\em freely independent} and hence the spectral distribution of $\bB^T_\bS A^T A \bB_\bS$ converges to a probability measure that can be described uniquely in terms of $\mu$ and $\mu'$. 

To characterize this measure, we use the following definition. The $R$-transform of a probability measure $\mu$ is given by
\begin{align*}
R_\mu(z) = S_\mu^{-1}(-z) - \frac{1}{z}
\end{align*}
where $S_\mu^{-1}(z)$ denotes the inverse (with respect to the composition of functions) of the Stieltjes transform,% (also known as the Cauchy transform):
\begin{align*}
S_\mu(z) = \int_{\mathbb{R}} \frac{1}{x-z}d\mu(x).
\end{align*}
The following result follows directly from Section 4.4 of Speicher's lecture on free probability \cite{speicher01}.

\begin{lemma}\label{lem:freeprob}
If the limiting spectral distribution of $A^T A$ is equal to $\tilde{\mu}$, then the limiting spectral distribution of the random matrix $\frac{1}{\Omega} \bB^T_\bS A^TA \bB_\bS$ is equal to $\tilde{\nu}$ almost surely where % whose $R$-transform is given by
\begin{align}\label{eq:freeprob}
R_{\tilde{\nu}}(z) = R_{\tilde{\mu}}(\Omega\, z).
\end{align}
%with
%\begin{align}
%\tilde{\mu} &= (1-\rho) \one(x) + \rho \mu\\
%\tilde{\nu} & = (1-\rho/\Omega) \one(x) + (\rho/\Omega) \nu
%\end{align}
\end{lemma}

From Lemma \ref{lem:freeprob} we conclude that spectral density of $A\bB_\bS \bB^T_\bS A^T$ converges to $\nu$ as $n \rw \infty$, where
$$ \nu = (\rho/\Omega -1) \delta_0  + (\Omega/\rho) \tilde{\nu}.$$ If $\nu$ is compactly supported then $ | \frac{1}{\Omega}A\bB_\bS \bB^T_\bS A^T|^{\frac{1}{m}}  \rw G_{\nu} $ almost surely as $n \rw \infty$. Thus we conclude that
\begin{align}
\limsup_{ n \rw \infty} I(A \bX;\bS|\bB) \le {\textstyle \frac{m}{2}} \log\left(\frac{1}{\theta(\Omega,F)} \cdot \frac{G_\mu}{G_\nu}\right)%^{\frac{1}{m}}}{ | \frac{1}{\Omega}A\bB_\bS \bB^T_\bS A^T|^{\frac{1}{m}}} \right)
\end{align}
almost surely for any compactly supported probability measures $\mu,\nu$ that satisfy Equation \eqref{eq:freeprob}. %One such choice is chose 
%Using the fact that the limiting spectra

Although the strongest bound corresponds to the minimization over $\mu$, such optimization appears to be difficult. Instead, we obtain a (potentially suboptimal) bound by setting $\mu$ equal to the Mar\u{c}enko-Pastur law \cite{MarcenkoPastur67} with parameter $\rho$, i.e.
 \begin{align*}
d\mu(x) = \frac{ \sqrt{(x-a)(b-x)}}{2 \pi \rho x}
\end{align*}
for all $x \in [a,b]$ where $a = (1-\sqrt{\rho})^2$ and $b= (1+\sqrt{\rho})^2$. Then, it can be shown that \eqref{eq:freeprob} is satisfied when $\nu$ is equal to the Mar\u{c}enko-Pastur law with parameter $\rho/\Omega$. Integrating with respect to these measures shows that
\begin{align*}
G_\mu &= e^{-1} \Delta(\rho)\\
G_\nu &= e^{-1} \Delta(\rho/\Omega)
\end{align*}
which completes the proof. 

We remark that convergence of spectral density to the Mar\u{c}enko-Pastur law corresponds to the setting where the elements of $A$ are i.i.d. zero mean Gaussian. Interestingly, it is possible to use the rotational invariance of the Gaussian distribution to obtain the same bound given above, without appealing to free probability. However, the approach taken above is more general and allows the calculation of the bound in terms of other limiting distributions. % in priallows for a potentially stronger bound.% if one is able to find a better distribution 
\section{Discussion}
Two insights from the field of compressed sensing are that any sparse vector can be sampled efficiently using linear projections, and that there exist random sampling matrices that good almost surely for any sparse basis. In this paper, we have investigated what happens if a probability measure is placed on the set of possible vectors and partial recovery is allowed by bounding the sampling rate distortion function $\rho(\alpha)$. In certain cases, we showed that the number of samples may be decreased. However, we also showed that in many cases, no reduction is possible, particularly if one requires universality with respect a sparse basis. %Additionally, we have used free probability to characterize certain random sub-matrices that occur frequently in compressed sensing. 
%universally good for any sparse basis. 
%
%using linear projections is well matched to the set of sparse vectors 

%
%%[HERE I WILL PUT SOME BREIF DISCUSSION DEPENDING ON TIME AND SPACE]
%One of the insights from compressed sensing is that sparsity and linear projections are well matched. The use of random projections give free universality with respect to basis. This paper stengthens these claims and shows that, in a certain sense, this is also the best that can be done. Discuss connections to the noisy setting here too.

\section*{Acknowledgment}
This work was supported in part by ARO MURI No. W911NF-06-1-0076.

\bibliographystyle{ieeetran}
\bibliography{gbib-long}

% Generated by IEEEtran.bst, version: 1.12 (2007/01/11)
\begin{thebibliography}{10}
\providecommand{\url}[1]{#1}
\csname url@samestyle\endcsname
\providecommand{\newblock}{\relax}
\providecommand{\bibinfo}[2]{#2}
\providecommand{\BIBentrySTDinterwordspacing}{\spaceskip=0pt\relax}
\providecommand{\BIBentryALTinterwordstretchfactor}{4}
\providecommand{\BIBentryALTinterwordspacing}{\spaceskip=\fontdimen2\font plus
\BIBentryALTinterwordstretchfactor\fontdimen3\font minus
  \fontdimen4\font\relax}
\providecommand{\BIBforeignlanguage}[2]{{%
\expandafter\ifx\csname l@#1\endcsname\relax
\typeout{** WARNING: IEEEtran.bst: No hyphenation pattern has been}%
\typeout{** loaded for the language `#1'. Using the pattern for}%
\typeout{** the default language instead.}%
\else
\language=\csname l@#1\endcsname
\fi
#2}}
\providecommand{\BIBdecl}{\relax}
\BIBdecl

\bibitem{riceweb}
``Compressive sensing resources,'' http://www.dsp.ece.rice.edu/cs.

\bibitem{Wainwright_allerton06}
M.~J. Wainwright, ``Sharp thresholds for high-dimensional and noisy recovery of
  sparsity,'' in \emph{Proc. Allerton Conf. on Comm., Control, and Computing},
  Monticello, IL, Sep. 2006.

\bibitem{Wainwright_isit07}
------, ``Information-theoretic bounds on sparsity recovery in the
  high-dimensional and noisy setting,'' in \emph{Proc. IEEE Int. Symp. on
  Inform. Theory}, Nice, France, Jun. 2007.

\bibitem{Reeves_masters}
G.~Reeves, ``Sparse signal sampling using noisy linear projections,'' Univ. of
  California, Berkeley, Dept. of Elec. Eng. and Cpmp. Sci., Tech. Rep.
  UCB/EECS-2008-3, Jan. 2008.

\bibitem{aeron-2007}
S.~Aeron, M.~Zhao, and V.~Saligrama, ``On sensing capacity of sensor networks
  for the class of linear observation, fixed snr models,'' Jun 2007,
  arXiv:0704.3434v3 [cs.IT].

\bibitem{aeron-2008}
------, ``Fundamental limits on sensing capacity for sensor networks and
  compressed sensing,'' Apr. 2008, arXiv:0804.3439v1 [cs.IT].

\bibitem{akcakaya-2007}
M.~Akcakaya and V.~Tarokh, ``Shannon theoretic limits on noisy compressive
  sampling,'' Nov. 2007, arXiv:0711.0366v1 [cs.IT].

\bibitem{RG08}
G.~Reeves and M.~Gastpar, ``Sampling bounds for sparse support recovery in the
  presence of noise,'' in \emph{Proc. IEEE Int. Symp. on Inform. Theory},
  Toronto, Canada, Jul. 2008.

\bibitem{fletcher-2008}
A.~K. Fletcher, S.~Rangan, and V.~K. Goyal, ``Necessary and sufficient
  conditions on sparsity pattern recovery,'' May 2008, arXiv:0804.1839v1
  [cs.IT].

\bibitem{WWR08}
W.~Wang, M.~J. Wainwright, and K.~Ramchandran, ``Information-theoretic limits
  on sparse signal recovery: Dense versus sparse measurement matrices,'' in
  \emph{Proc. IEEE Int. Symp. on Inform. Theory}, Toronto, Canada, Jul. 2008.

\bibitem{RG-Lower-Bounds}
G.~Reeves and M.~Gastpar, ``Approximate sparsity pattern recovery:
  Information-theoretic lower bounds,'' Preprint.

\bibitem{RG-Upper-Bounds}
------, ``Approximate sparsity pattern recovery: Information-theoretic upper
  bounds,'' Preprint.

\bibitem{RG_asilomar09}
------, ``A note on optimal support recovery in a note on optimal support
  recovery in compressed sensing,'' in \emph{Proc. IEEE Asilomar Conference on
  Signals, Systems, and Computers}, Monterey, CA, Nov. 2009.

\bibitem{elementsofIT}
T.~M. Cover and J.~A. Thomas, \emph{Elements of Information Theory}.\hskip 1em
  plus 0.5em minus 0.4em\relax New York: Wiley, 1991.

\bibitem{Voicu83}
D.~Voiculescu, ``Asymptotically commuting finite rank unitary operators without
  commuting approximants,'' \emph{Acta Sci. Math.}, vol.~45, pp. 429--432,
  1983.

\bibitem{speicher01}
R.~Speicher, ``Free probability theory and random matrices,'' lectures at the
  summer school on 'Asymptotic Combinatorics with Application to Mathematical
  Physics', Jul. 2001.

\bibitem{MarcenkoPastur67}
V.~Mar\u{c}enko and L.~{Pastur}, ``{Distribution of eigenvalues for some sets
  of random matrices},'' \emph{Math. USSR-Sbornik}, vol.~1, pp. 457--483, 1967.

\end{thebibliography}

\end{document}